\documentclass[onecolumn,amsmath,amssymb,10pt,aps]{revtex4}
  
\usepackage[utf8]{inputenc}
\usepackage[T1]{fontenc}

\usepackage{amsmath}
\usepackage{amsfonts}
\usepackage{graphicx}
\usepackage{yfonts}
\usepackage{color}
\usepackage[normalem]{ulem}
\usepackage{amsthm}
\usepackage{bm}
\usepackage{bbm}
\usepackage{mathtools}
\usepackage{array}
\usepackage{placeins}
\usepackage{enumitem}
\usepackage{tikz}
\usepackage{multirow}
\usepackage{ulem}

\def\<{\langle}
\def\>{\rangle}

\newcommand{\Tr}{\mathrm{Tr}}
\def\oper{{\mathchoice{\rm 1\mskip-4mu l}{\rm 1\mskip-4mu l}
{\rm 1\mskip-4.5mu l}{\rm 1\mskip-5mu l}}}
\DeclareMathAlphabet\mathbfcal{OMS}{cmsy}{b}{n}

\mathchardef\mhyphen="2D 

\newtheorem{Lemma}{Lemma}

\newtheorem{Definition}{Definition}
\newtheorem{Remark}{Remark}
\newtheorem{Proposition}{Proposition}
\newtheorem{Example}{Example}

\begin{document}

\title{Families of $k$-positive maps and Schmidt number witnesses from generalized equiangular measurements}

\author{Katarzyna Siudzi\'{n}ska$^\ast$}
\affiliation{Institute of Physics, Faculty of Physics, Astronomy and Informatics \\  Nicolaus Copernicus University in Toru\'{n}, ul. Grudzi\k{a}dzka 5, 87--100 Toru\'{n}, Poland \\ e-mail: kasias@umk.pl}

\begin{abstract}
Quantum entanglement is an important resource in many modern technologies, like 
quantum computation or quantum communication and information processing. Therefore, most interest is given to detect and quantify entangled states. Entanglement degree of bipartite mixed quantum states can be measured using the Schmidt number. Witnesses of the Schmidt number are closely related to $k$-positive linear maps, for which there is no general construction. Here, we use the generalized equiangular measurements to define a family of $k$-positive linear maps and the corresponding Schmidt number witnesses. We present examples of witnesses that detect two-parameter entangled states with Schmidt number $k+1$ in any dimension. 
Our approach allows for a more efficient entanglement quantification than other known Schmidt number witnesses constructed from symmetric measurement operators.
\end{abstract}

\flushbottom

\maketitle

\thispagestyle{empty}

\section{Introduction}

Quantum entanglement is an important resource in quantum computation and quantum information processing. It finds numerous applications e.g. in quantum teleportation \cite{Masanes}, cryptography \cite{Ekert}, dense coding \cite{Wiesner}, and channel discrimination \cite{BaeDarek}. Therefore, much interest is given to characterize the degree of entanglement of a quantum state. A popular entanglement measure is the Schmidt number, which determines the minimal Schmidt rank of the pure states that construct the density operator \cite{Guhne,Terhal}. The Schmidt number criteria have been constructed using the covariance matrix \cite{Liu4,Liu5}, cross norm \cite{Johnston4}, and Bloch decomposition \cite{Klockl}.

Entanglement witnesses prove to be practical tools in the detection of entangled states, as they allow for entanglement quantification without the full tomography of a quantum state. From definition, a Hermitian operator $W$ is an entanglement witness provided that its expectation value $\Tr(\rho W)\geq 0$ on any separable state $\rho$ and there exists an entangled state $\rho$ such that $\Tr(\rho W)<0$ \cite{Terhal1,Terhal2}. The construction of entanglement witnesses is based on positive but not completely positive linear maps and can be extended to the Schmidt number witnesses, where $k$-positive maps play a crucial role. A Hermitian operator $W_k$ detects the Schmidt number $k+1$ of a quantum state $\rho$ if $\Tr(W_k\rho)\geq 0$ for all states with the Schmidt number up to $k$ and there exists a state $\rho$ for which $\Tr(W_k\rho)<0$ \cite{Terhal}. The Schmidt number witnesses have been used to quantify entanglement degree of high dimensional and bound entangled states \cite{Hulpke,Sanpera,Wyderka}.

Recently, much interest has been given to entanglement quantification using quantum measurements. In quantum information, they are represented by positive, operator-valued measures (POVMs), which are sets of positive operators that sum up to the identity operator. Sufficient Schmidt number criteria have been obtained using the correlation matrix constructed from the symmetric, informationally complete (SIC) POVMs and mutually unbiased bases (MUBs) \cite{Morelli}, general SIC POVMs \cite{Wangs}, $(N,M)$-POVMs \cite{Schmidt_NM}, and generalized equiangular measurements \cite{GEAM_coherence}. The Schmidt number $k+1$ witnesses have been constructed from $k$-positive maps for SIC POVMs, MUBs \cite{Shi}, and $(N,M)$-POVMs \cite{SIC-MUB_kpos}. For $k=1$, one recovers many classes of entanglement witnesses: from mutually unbiased bases \cite{MUBs,EW-2MUB,bound_ent} and measurements \cite{Li,P_maps,MUM_purity}, SIC POVMs \cite{EW-SIC}, symmetric measurements \cite{SIC-MUB}, and equiangular tight frames \cite{W_ETS,Rastegin_EW}.

In this paper, we propose a wide family of $k$-positive (but not trace preserving) maps associated with generalized equiangular measurements (GEAMs). These are informationally overcomplete measurements that include sets of mutually unbiased equiangular tight frames of abitrary rank \cite{GEAM}. Interestingly, the $k$-positivity of the map is encoded into a single scalar parameter that depends only on three constants that characterize the measurement. Moreover, we show that, in the special cases where other constuctions are known, our maps are less positive for any fixed $k$ and therefore allow for a better Schmidt number detection with the corresponding Schmidt number witness. Finally, we provide examples of witnesses for two-parameter states in high-dimensional scenarios.

\section{Generalized equiangular measurements}

Quantum measurements are represented by positive, operator-valued measures (POVMs) $\mathcal{P}=\{P_k,\,k=1,\ldots,M\}$, which are semi-positive operators ($P_k\geq 0$) acting on the Hilbert space $\mathcal{H}\simeq\mathbb{C}^d$ that sum up to the identity operator ($\sum_{k=1}^MP_k=\mathbb{I}_d$). Projective measurements are POVMs such that all $P_k$ are rank-1 projectors rescaled by positive constants. Additional symmetry conditions can be introduced by requiring equiangularity between different operator pairs. This means that $P_k$ satisfy the trace condition $\Tr(P_kP_\ell)=c\Tr(P_k)\Tr(P_\ell)$ for all $k\neq\ell$ \cite{EOM22}. Such POVMs are known as equiangular measurements. If instead $\sum_{k=1}^MP_k=\gamma\mathbb{I}_d$ for some positive number $\gamma$, then $\mathcal{P}$ is an equiangular tight frame \cite{Strohmer}. The notion of equiangular measurements has been generalized by considering informationally complete collections of mutually unbiased equiangular tight frames and relaxing the requirement that $P_k$ are rank-1.

\begin{Definition}\label{geam}(\cite{GEAM})
A collection $\mathcal{P}=\cup_{\alpha=1}^N\mathcal{P}_\alpha$ of $N$ generalized equiangular tight frames $\mathcal{P}_\alpha=\{P_{\alpha,k}:\,k=1,\ldots,M_\alpha\}$ is the generalized equiangular measurement (GEAM) if
\begin{enumerate}[label=(\arabic*)]
\item $\sum_{k=1}^{M_\alpha}P_{\alpha,k}=\gamma_\alpha\mathbb{I}_d$, where $\gamma_\alpha>0$ and $\sum_{\alpha=1}^N\gamma_\alpha=1$;
\item it has $|\mathcal{P}|=\sum_{\alpha=1}^NM_\alpha=d^2+N-1$ elements;
\item the following trace conditions hold,
\begin{equation}
\begin{split}
\Tr(P_{\alpha,k})&=a_\alpha,\\
\Tr(P_{\alpha,k}^2)&=b_\alpha \Tr(P_{\alpha,k})^2,\\
\Tr(P_{\alpha,k}P_{\alpha,\ell})&=c_\alpha
\Tr(P_{\alpha,k})\Tr(P_{\alpha,\ell}),\qquad k\neq\ell,\\
\Tr(P_{\alpha,k}P_{\beta,\ell})&=
f\Tr(P_{\alpha,k})\Tr(P_{\beta,\ell}),\qquad \alpha\neq\beta,
\end{split}
\end{equation}
with the symmetry parameters
\begin{equation}
a_\alpha=\frac{d\gamma_\alpha}{M_\alpha},\qquad c_\alpha=\frac{M_\alpha-db_\alpha}{d(M_\alpha-1)},
\qquad f=\frac 1d,
\end{equation}
\begin{equation}\label{ba}
\frac 1d <b_\alpha\leq\frac 1d \min\{d,M_\alpha\}.
\end{equation}
\end{enumerate}
\end{Definition}

It is important to note that many well-known classes of operators are GEAMs, including general symmetric, informationally complete (SIC) POVMs \cite{Gour}, mutually unbiased measurements \cite{Kalev}, $(N,M)$-POVMs \cite{SIC-MUB}, and generalized symmetic measurements \cite{SIC-MUB_general}.

Any generalized equiangular measurement can be constructed from a Hermitian orthonormal operator basis 
\begin{equation}\label{Ga}
\mathcal{G}=\{G_0=\mathbb{I}_d/\sqrt{d},G_{\alpha,k}:\,k=1,\ldots,M_\alpha-1;\,\alpha=1,\ldots,N\}
\end{equation}
with traceless $G_{\alpha,k}$. Indeed, one has \cite{GEAM}
\begin{equation}\label{Pak}
P_{\alpha,k}=\frac{a_\alpha}{d}\mathbb{I}_d+\tau_\alpha H_{\alpha,k},
\end{equation}
where
\begin{equation}\label{H}
H_{\alpha,k}=\left\{\begin{aligned}
&G_\alpha-\sqrt{M_\alpha}(1+\sqrt{M_\alpha})G_{\alpha,k},\quad k=1,\ldots,M_\alpha-1,\\
&(1+\sqrt{M_\alpha})G_\alpha,\qquad k=M_\alpha,
\end{aligned}\right.
\end{equation}
$G_\alpha=\sum_{k=1}^{M_\alpha-1}G_{\alpha,k}$, and the real parameters
\begin{equation}
\tau_\alpha=\pm\sqrt{\frac{S_\alpha}{M_\alpha(\sqrt{M_\alpha}+1)^2}},\qquad
S_\alpha=a_\alpha^2(b_\alpha-c_\alpha),
\end{equation}
are fixed to guarantee that $P_{\alpha,k}\geq 0$. 
If $S_\alpha=a_\alpha^2(b_\alpha-c_\alpha)\equiv S$ for all $\alpha=1,\ldots,N$, then the generalized equiangular measurements are equidistant in the Frobenius norm \cite{conical},
\begin{equation}
D_2^2(P_{\alpha,k},P_{\alpha,\ell})=\frac 12\|P_{\alpha,k}-P_{\alpha,\ell}\|_2^2
=\frac 12\Tr[(P_{\alpha,k}-P_{\alpha,\ell})^2]=S,
\end{equation}
where the admissible range of $S$ reads \cite{GEAM}
\begin{equation}\label{Srange}
0<S\leq \min_\alpha\left\{\frac{d\gamma_\alpha}{M_\alpha},
\frac{d-1}{M_\alpha-1}\frac{d\gamma_\alpha}{M_\alpha}\right\}.
\end{equation}

{It has been shown that equidistant GEAMs are a special case of a more general concept -- namely, conical 2-design \cite{Graydon,Graydon2}. Introduced as a generalization of complex projective designs, conical designs are semi-positive operators $E_k$ of arbitary rank such that $\sum_kE_k\otimes E_k$ commutes with $U\otimes U$. Notably, $E_k$ span the operator space on $\mathbb{C}^d$, and 
$U$ is an arbitrary unitary operator on $\mathbb{C}^d$. Now, equidistant GEAMs form a conical 2-design, which means that \cite{GEAM}
\begin{equation}\label{con}
\sum_{\alpha=1}^N\sum_{k=1}^{M_\alpha}P_{\alpha,k}\otimes P_{\alpha,k}=
\kappa_+\mathbb{I}_d\otimes\mathbb{I}_d+\kappa_-\mathbb{F}_d,
\end{equation}
with
\begin{equation}\label{kappas2}
\kappa_+=\mu_N-\frac{S}{d},\qquad\kappa_-=S,\qquad {\rm where}\qquad \mu_N=\frac 1d \sum_{\alpha=1}^Na_\alpha\gamma_\alpha,
\end{equation}
and $\mathbb{F}_d=\sum_{m,n=0}^{d-1}|m\>\<n|\otimes|n\>\<m|$ denotes the flip operator. For any conical 2-design, the parameters $\kappa_\pm$ necessarily satisfy $\kappa_+\geq\kappa_->0$. It is easy to check that this holds for $\kappa_\pm$ from eq. (\ref{kappas2}), as equivalently $S\leq d\mu_N/(d+1)$, which is a direct consequence of eq. (\ref{Srange}) \cite{GEAM}.}

Additionally, only for the GEAMs that are equidistant, there exists a linear relation between the purity $\Tr(\rho^2)$ and the index of coincidence \cite{Rastegin5}
\begin{equation}\label{pak}
\mathcal{C}(\rho)=\sum_{\alpha=1}^N\sum_{k=1}^{M_\alpha}p_{\alpha,k}^2,\qquad p_{\alpha,k}=\Tr(P_{\alpha,k}\rho).
\end{equation}
Indeed, one has \cite{GEAM}
\begin{equation}\label{IOCN}
\mathcal{C}(\rho)=S\left(\Tr\rho^2-\frac 1d\right)+\mu_N.
\end{equation}
In ref. \cite{GEAM_Pmaps}, this result has been generalized to a partial index of coincidence
\begin{equation}\label{IOCL}
\mathcal{C}_L=\sum_{\alpha=1}^L\sum_{k=1}^{M_\alpha}p_{\alpha,k}^2\leq S\left(\Tr\rho^2-\frac 1d\right)+\mu_L,
\end{equation}
where $\mu_L=(1/d)\sum_{\alpha=1}^La_\alpha\gamma_\alpha$ and $1\leq L\leq N$.
An analogical relation holds after replacing a quantum state $\rho$ with a linear operator $X$.


\begin{Lemma}\label{Lemma}
For any equidistant GEAM and a linear operator $X$,
\begin{equation}\label{IOC}
\sum_{\alpha=1}^L\sum_{k=1}^{M_\alpha}|\Tr(P_{\alpha,k}X)|^2\leq 
S\Tr(X^\dagger X)+|\Tr(X)|^2\left(\mu_L-\frac{S}{d}\right),
\end{equation}
where $\mu_L=(1/d)\sum_{\alpha=1}^La_\alpha\gamma_\alpha$ and $1\leq L\leq N$. The equality is reached for $L=N$.
\end{Lemma}

The proof is given in Appendix \ref{AppA}.

Note that generalized equiangular measurements are closely related to mutually unbiased equiangular tight frames (MU GETFs) \cite{conical}. Indeed, if $P_{\alpha,k}$ form a GEAM, then $\mu P_{\alpha,k}$ with $\mu>0$ defines a maximal set of MU GETFs. We end this section with another important observation.

\begin{Proposition}
Any GEAM can be rescaled to become equidistant.
\end{Proposition}

\begin{proof}
Take a GEAM $\mathcal{P}=\{P_{\alpha,k}:\,k=1,\ldots,M_\alpha;\,\alpha=1,\ldots,N\}$ for which $S_\alpha=a_\alpha^2(b_\alpha-c_\alpha)\neq S_\beta$, $\beta\neq\alpha$. Using its elements, construct $\widetilde{P}_{\alpha,k}=\mu_\alpha P_{\alpha,k}$ with positive parameters $\mu_\alpha$. Observe that $\widetilde{\mathcal{P}}=\{\widetilde{P}_{\alpha,k}:\,k=1,\ldots,M_\alpha;\,\alpha=1,\ldots,N\}$ is also a GEAM as long as $\sum_{\alpha=1}^N\mu_\alpha\gamma_\alpha=1$, and moreover
\begin{equation}
\begin{split}
\Tr(\widetilde{P}_{\alpha,k})&=\widetilde{a}_\alpha=\mu_\alpha a_\alpha,\\
\Tr(\widetilde{P}_{\alpha,k}^2)&=\widetilde{a}_\alpha^2b_\alpha,\\
\Tr(\widetilde{P}_{\alpha,k}\widetilde{P}_{\alpha,\ell})&=\widetilde{a}_\alpha^2c_\alpha,\qquad k\neq\ell,\\
\Tr(\widetilde{P}_{\alpha,k}\widetilde{P}_{\beta,\ell})&=
\widetilde{a}_\alpha\widetilde{a}_\beta f,\qquad \alpha\neq\beta.
\end{split}
\end{equation}
Hence, a simple rescaling of the measurement operators changes only one family of the symmetry parameters: $a_\alpha\to\widetilde{a}_\alpha=\mu_\alpha a_\alpha$, leaving $b_\alpha$, $c_\alpha$, and $f$ unchanged. Now, taking
\begin{equation}
\mu_\alpha=\frac{\eta}{\sqrt{S_\alpha}},\qquad \eta=\left(\sum_{\alpha=1}^N\frac{\gamma_\alpha}{\sqrt{S_\alpha}}\right)^{-1},
\end{equation}
it is straightforward to check that indeed $\sum_{\alpha=1}^N\mu_\alpha\gamma_\alpha=1$ and
\begin{equation}
\widetilde{S}_\alpha=\widetilde{a}_\alpha^2(b_\alpha-c_\alpha)=\mu_\alpha^2 S_\alpha=\eta^2\equiv\widetilde{S}
\end{equation}
for all $\alpha=1,\ldots,N$.
\end{proof}

\section{Construction of $k$-positive maps}

A linear map $\Phi:\mathcal{B}(\mathcal{H})\to\mathcal{B}(\mathcal{H})$ for $\mathcal{H}\simeq\mathbb{C}^d$ is $k$-positive if $\oper_k\otimes\Phi[X]\geq 0$ for all $X\geq 0$, where $\oper_k$ is the identity map on the space of $k\times k$ complex matrices \cite{Marciniak,Ende2}. The integer $k$ varies from $k=1$ for positive to $k=d$ for completely positive maps \cite{Stormer,Stormer2,Paulsen}. Actually, there is an inclusion hierarchy between the sets $\mathcal{L}_k$ of $k$-positive maps: $\mathcal{L}_{k+1}\subset\mathcal{L}_k$ for $k=1,\ldots,d-1$.

{In this section, we use the generalized equiangular measurements to construct families of $k$-positive linear maps.}

{\begin{Definition}
Take the generalized equiangular measurement (GEAM) $\{P_{\alpha,k}:\,k=1,\ldots,M_\alpha;\,\alpha=1,\ldots,N\}$ from Definition \ref{geam} that is equidistant -- that is, its parameters satisfy $a_\alpha^2(b_\alpha-c_\alpha)\equiv S$ for all $\alpha=1,\ldots,N$. Using the equidistant GEAM, define $N$ linear maps
\begin{equation}
\Phi_\alpha[X]=\sum_{k,\ell=1}^{M_\alpha}\mathcal{O}^{(\alpha)}_{k\ell}P_{\alpha,k}\Tr(XP_{\alpha,\ell}),
\end{equation}
where $\mathcal{O}^{(\alpha)}$ are orthogonal rotation matrices such that $\mathcal{O}^{(\alpha)}\mathbf{n}_\ast=\mathbf{n}_\ast$ for $\mathbf{n}_\ast=(1,\ldots,1)$.
\end{Definition}}

Observe that, in general, $\Phi_\alpha$ are not trace preserving due to a rescaling factor,
\begin{equation}
\Tr(\Phi_\alpha[X])=a_\alpha\gamma_\alpha\Tr(X).
\end{equation}
For $\mathcal{O}^{(\alpha)}=\mathbb{I}_{M_\alpha}$, one recovers the Holevo-like form \cite{EBC}, and hence $\Phi_\alpha$ is then completely positive and entanglement breaking. Now, if we take a sum of some $\pm\Phi_\alpha$ and the completely depolarizing channel $\Phi_0[X]=\mathbb{I}_d\Tr(X)/d$ rescaled by a (not necessarily positive) real factor, then the resulting map is $k$-positive under certain conditions.

\begin{Proposition}\label{Prop}
The linear map
\begin{equation}\label{kpos}
\Phi^{(k)}=A_k\Phi_0+\sum_{\alpha=L+1}^K\Phi_\alpha-\sum_{\alpha=1}^L\Phi_\alpha
\end{equation}
with {$A_k=d(2\mu_L-\mu_K)+(dk-1)S$} and $1\leq L\leq K\leq N$ is $k$-positive.
\end{Proposition}

\begin{proof}
{From Mehta's theorem \cite{Mehta} (see also Section 8.4 in ref. \cite{Zyczkowski}), a map $\Lambda:\mathbb{C}^d\to\mathbb{C}^d$ is positive if
\begin{equation}\label{pos}
\frac{\Tr[(\Lambda[P])^2]}{[\Tr(\Lambda[P])]^2}\leq\frac{1}{d-1}
\end{equation}
for any rank-1 projector $P$. Therefore, a sufficient condition for the $k$-positivity of $\Phi^{(k)}$ follows from eq. (\ref{pos}) if one takes $\Lambda=\oper_k\otimes\Phi^{(k)}$ and $P=P_k$, where $P_k$ is a projector onto the maximally entangled vector with Schmidt rank $k$. Now, the sufficient positivity condition becomes
\begin{equation}\label{pos2}
\frac{\Tr[(\oper_k\otimes\Phi^{(k)}[P_k])^2]}{[\Tr(\oper_k\otimes\Phi^{(k)}[P_k])]^2}\leq\frac{1}{dk-1}.
\end{equation}}
In the most general scenario, 
\begin{equation}\label{Pk}
P_k=(U\otimes V)P_+^{(k)}(U\otimes V)^\dagger,\qquad P_+^{(k)}=\frac 1k \sum_{m,n=0}^{k-1} |m\>\<n|\otimes|m\>\<n|,
\end{equation}
where $U$ and $V$ are arbitrary unitary operators acting locally on the subsystems. The value of the denominator
\begin{equation}
[\Tr(\oper_k\otimes\Phi^{(k)}[P_k])]^2=[A_k-d(2\mu_L-\mu_K)]^2
\end{equation}
follows from straightforward calculations,
\begin{equation}
\begin{split}
\Tr(\oper_k\otimes\Phi^{(k)}[P_k])&=\Tr(\oper_k\otimes\Phi^{(k)}[P_+^{(k)}])
=\frac 1k \sum_{m=0}^{k-1} \Tr(\Phi^{(k)}[|m\>\<m|])
=\frac 1k \sum_{m=0}^{k-1}[A_k-d(2\mu_L-\mu_K)]\Tr(|m\>\<m|)\\&
=A_k-d(2\mu_L-\mu_K).
\end{split}
\end{equation}

It remains to compute the nominator. First, observe that
{
\begin{equation}\label{P0}
\Tr[(\oper_k\otimes\Phi^{(k)}[P_k])^2]=\frac {1}{k^2} \sum_{m,n=0}^{k-1}\Tr(\Phi^{(k)}[V_{mn}]\Phi^{(k)}[V_{nm}]),
\end{equation}}
where we have introduced the family of operators $V_{mn}=V|m\>\<n|V^\dagger$ using the same symbols as in eq. (\ref{Pk}). The expression under the sum can be expanded into
\begin{equation}\label{P1}
\begin{split}
\Tr(\Phi^{(k)}[V_{mn}]&\Phi^{(k)}[V_{nm}])=A_k^2\Tr(\Phi_0[V_{mn}]\Phi_0[V_{nm}])
+A_k\Bigg\{\sum_{\beta=L+1}^K\Tr(\Phi_0[V_{mn}]\Phi_\beta[V_{nm}])
-\sum_{\beta=1}^L\Tr(\Phi_0[V_{mn}]\Phi_\beta[V_{nm}])\\
&+\sum_{\alpha=L+1}^K\Tr(\Phi_\alpha[V_{mn}]\Phi_0[V_{nm}])
-\sum_{\alpha=1}^L\Tr(\Phi_\alpha[V_{mn}]\Phi_0[V_{nm}])\Bigg\}
+\sum_{\alpha,\beta=L+1}^K\Tr(\Phi_\alpha[V_{mn}]\Phi_\beta[V_{nm}])\\
&+\sum_{\alpha,\beta=1}^L\Tr(\Phi_\alpha[V_{mn}]\Phi_\beta[V_{nm}])
-\sum_{\alpha=1}^L\sum_{\beta=L+1}^K\Tr(\Phi_\alpha[V_{mn}]\Phi_\beta[V_{nm}])
-\sum_{\beta=1}^L\sum_{\alpha=L+1}^K\Tr(\Phi_\alpha[V_{mn}]\Phi_\beta[V_{nm}]).
\end{split}
\end{equation}
Further simplification is possible after implementation of the properties of the rotation matrices $\mathcal{O}^{(\alpha)}$;
\begin{equation}\label{rot}
\sum_{j=1}^{M_\alpha}\mathcal{O}_{j\ell}^{(\alpha)}=
\sum_{\ell=1}^{M_\alpha}\mathcal{O}_{j\ell}^{(\alpha)}=1,\qquad
\sum_{r=1}^{M_\alpha}\mathcal{O}_{jr}^{(\alpha)}
\mathcal{O}_{\ell r}^{(\alpha)}=\delta_{j\ell}.
\end{equation}
From the definitions of $\Phi_0$ and $\Phi_\alpha$, we find
\begin{equation}
\begin{split}
\Tr(\Phi_0[V_{mn}]\Phi_0[V_{nm}])&=\frac 1d \delta_{mn},\\
\Tr(\Phi_0[V_{mn}]\Phi_\alpha[V_{nm}])&=\frac 1d a_\alpha\gamma_\alpha \delta_{mn},\\
\Tr(\Phi_\alpha[V_{mn}]\Phi_\beta[V_{nm}])&=\frac 1d a_\alpha a_\beta \gamma_\alpha\gamma_\beta\delta_{mn},\qquad\beta\neq\alpha,\\
\Tr(\Phi_\alpha[V_{mn}]\Phi_\alpha[V_{nm}])&=a_\alpha^2\gamma_\alpha^2c_\alpha\delta_{mn}+ S\sum_{\ell=1}^{M_\alpha}|\Tr(P_{\alpha,\ell}V_{mn})|^2.
\end{split}
\end{equation}
Inputting the first three formulas into eq. (\ref{P1}) results in
\begin{equation}\label{P2}
\begin{split}
\Tr(\Phi^{(k)}[V_{mn}]\Phi^{(k)}[V_{nm}])&=\frac 1d A_k^2 \delta_{mn}
+\frac 2d A_k\delta_{mn}
\Bigg[\sum_{\alpha=L+1}^Ka_\alpha\gamma_\alpha
-\sum_{\alpha=1}^La_\alpha\gamma_\alpha\Bigg]\\
&+\sum_{\alpha,\beta=L+1}^K\Tr(\Phi_\alpha[V_{mn}]\Phi_\beta[V_{nm}])
+\sum_{\alpha,\beta=1}^L\Tr(\Phi_\alpha[V_{mn}]\Phi_\beta[V_{nm}])\\
&-\frac 2d\delta_{mn}\sum_{\alpha=1}^La_\alpha\gamma_\alpha\sum_{\beta=L+1}^K  a_\beta\gamma_\beta\\
&=\frac 1d \delta_{mn} \Big[A_k^2 +2d A_k(\mu_K-2\mu_L)
-2d^2\mu_L(\mu_K-\mu_L)\Big]\\
&+\sum_{\alpha,\beta=L+1}^K\Tr(\Phi_\alpha[V_{mn}]\Phi_\beta[V_{nm}])
+\sum_{\alpha,\beta=1}^L\Tr(\Phi_\alpha[V_{mn}]\Phi_\beta[V_{nm}]).
\end{split}
\end{equation}
The fourth formula allows us to rewrite the remaining sums, so that
\begin{equation}
\begin{split}
&\sum_{\alpha,\beta=L+1}^K\Tr(\Phi_\alpha[V_{mn}]\Phi_\beta[V_{nm}])
+\sum_{\alpha,\beta=1}^L\Tr(\Phi_\alpha[V_{mn}]\Phi_\beta[V_{nm}])\\
&=\sum_{\alpha=1}^K\Tr(\Phi_\alpha[V_{mn}]\Phi_\alpha[V_{nm}])
+\sum_{\alpha=L+1}^K\sum_{\substack{\beta=L+1\\\beta\neq\alpha}}^K\Tr(\Phi_\alpha[V_{mn}]\Phi_\beta[V_{nm}])
+\sum_{\alpha=1}^L\sum_{\substack{\beta=1\\\beta\neq\alpha}}^L\Tr(\Phi_\alpha[V_{mn}]\Phi_\beta[V_{nm}])\\
&=\sum_{\alpha=1}^K\Tr(\Phi_\alpha[V_{mn}]\Phi_\alpha[V_{nm}])
+\frac 1d \delta_{mn} \sum_{\alpha=1}^L\sum_{\substack{\beta=1\\\beta\neq\alpha}}^La_\alpha a_\beta \gamma_\alpha\gamma_\beta
+\frac 1d \delta_{mn}
\sum_{\alpha=L+1}^K\sum_{\substack{\beta=L+1\\\beta\neq\alpha}}^Ka_\alpha a_\beta \gamma_\alpha\gamma_\beta\\
&=\sum_{\alpha=1}^K\Tr(\Phi_\alpha[V_{mn}]\Phi_\alpha[V_{nm}])
+\frac 1d \delta_{mn}\left[d^2\mu_L^2+d^2(\mu_K-\mu_L)^2-\sum_{\alpha=1}^K a_\alpha^2\gamma_\alpha^2\right]\\
&=\frac 1d \delta_{mn}\left[d^2\mu_L^2+d^2(\mu_K-\mu_L)^2+d\sum_{\alpha=1}^K a_\alpha^2\gamma_\alpha^2(c_\alpha-f)\right]
+S\sum_{\alpha=1}^K\sum_{\ell=1}^{M_\alpha}|\Tr(P_{\alpha,\ell}V_{mn})|^2.
\end{split}
\end{equation}
Finally, we find that eq. (\ref{P2}) reduces to
\begin{equation}\label{P3}
\begin{split}
\Tr(\Phi^{(k)}[V_{mn}]\Phi^{(k)}[V_{nm}])&=\frac 1d \delta_{mn} \Big[A_k+d (\mu_K-2\mu_L)\Big]^2+\delta_{mn}\sum_{\alpha=1}^Ka_\alpha^2\gamma_\alpha^2(c_\alpha-f)
+S\sum_{\alpha=1}^K\sum_{\ell=1}^{M_\alpha}|\Tr(P_{\alpha,\ell}V_{mn})|^2\\
&=\frac 1d \delta_{mn} \Big[A_k+d (\mu_K-2\mu_L)\Big]^2
+S\left[\sum_{\alpha=1}^K\sum_{\ell=1}^{M_\alpha}|\Tr(P_{\alpha,\ell}V_{mn})|^2
-\mu_K\delta_{mn}\right],
\end{split}
\end{equation}
where we used $c_\alpha-f=-(b_\alpha-c_\alpha)/M_\alpha$ \cite{GEAM}. From Lemma 1, we get
\begin{equation}
\sum_{\alpha=1}^K\sum_{\ell=1}^{M_\alpha}|\Tr(P_{\alpha,\ell}V_{mn})|^2\leq 
S+\left(\mu_K-\frac{S}{d}\right)\delta_{mn}.
\end{equation}
Now, going back to eq. (\ref{P0}), we obtain the upper bound for the nominator,
{\begin{equation}\label{P0F}
\begin{split}
\Tr[(\oper_k\otimes\Phi^{(k)}[P_k])^2]&=\frac {1}{k^2} \sum_{m,n=0}^{k-1}\Tr(\Phi^{(k)}[V_{mn}]\Phi^{(k)}[V_{nm}])\\
&=\frac {1}{k^2} \sum_{m,n=0}^{k-1}
\left\{\frac 1d \delta_{mn} \Big[A_k+d (\mu_K-2\mu_L)\Big]^2
+S\left[\sum_{\alpha=1}^K\sum_{\ell=1}^{M_\alpha}|\Tr(P_{\alpha,\ell}V_{mn})|^2
-\mu_K\delta_{mn}\right]\right\}\\
&\leq\frac {1}{dk} \Big[A_k-d (2\mu_L-\mu_K)\Big]^2+\frac{S^2}{dk}(dk-1).
\end{split}
\end{equation}}
Therefore, the sufficient $k$-divisibility condition from {eq. (\ref{pos2})} reads
{\begin{equation}\label{posF}
\frac{\Big[A_k-d (2\mu_L-\mu_K)\Big]^2+S^2(dk-1)}{dk\Big[A_k-d (2\mu_L-\mu_K)\Big]^2}\leq\frac{1}{dk-1}.
\end{equation}}
The equality holds provided that $A_k$ is a solution of the quadratic equation
{\begin{equation}
\Big[A_k-d (2\mu_L-\mu_K)\Big]^2=(dk-1)^2S^2.
\end{equation}}
For $\Tr(\oper_k\otimes\Phi^{(k)}[P_k])=A_k-d (2\mu_L-\mu_K)>0$, one finds
{\begin{equation}
A_k=d(2\mu_L-\mu_K)+(dk-1)S.
\end{equation}}
\end{proof}

Observe that the entire $k$-dependence of $\Phi^{(k)}$ is in the scalar parameter $A_k$. Moreover, $A_k$ also depends on two measurement parameters $\mu_K$, $\mu_L$: one connected to the number $L$ of negative terms and the other to the number $K$ of all terms. For $k=1$, $A_1=d(2\mu_L-\mu_K)+(d-1)S$ and one recovers the sufficient positivity conditions from ref. \cite{GEAM_Pmaps}. For $k=d$, {$A_d=d (2\mu_L-\mu_K)+(d^2-1)S$}, and our results provide sufficient conditions for complete positivity of $\Phi^{(d)}$, which means it can be used as a quantum channel after a simple rescaling
by a factor {$1/[(d^2-1)S]$} to guarantee its trace preservation.

{\begin{Remark}
The $k$-positive maps form a convex set, so if $\Phi^{(k)}$ is $k$-positive, then any map
\begin{equation}
\Psi^{(k)}=(A_k+\mu_0)\Phi_0+\sum_{\alpha=L+1}^K\mu_\alpha\Phi_\alpha-\sum_{\alpha=1}^L\mu_\alpha\Phi_\alpha
\end{equation}
with $\mu_0\geq 0$, $\mu_\alpha\geq 1$, and $\mu_\beta\leq 1$ ($\alpha=L+1,\ldots,K$, $\beta=1,\ldots,L$) is also $k$-positive.
\end{Remark}}

A class of $k$-positive maps has been introduced in ref. \cite{SIC-MUB_kpos} for informationally complete $(N,M)$-POVMs, which are a special class of GEAMs for $\gamma_\alpha=1/N$ and $M_\alpha=M$, $\alpha=1,\ldots,N$. After rescaling by a positive factor, these maps read
\begin{equation}
\Lambda_k=B_k\Phi_0-\sum_{\alpha=1}^N\Phi_\alpha,\qquad B_k=N^2\mu_N\left[d+(dk-1)\sqrt{S\left(d+(d^2-1)\frac{S}{\mu_N}\right)}\right].
\end{equation}
Note that there is an additional dependence on the parameter $N$. For comparison, eq. (\ref{kpos}) for $(N,M)$-POVMs and $L=K=N$ reduces to
{
\begin{equation}
\Phi^{(k)}=A_k\Phi_0-\sum_{\alpha=1}^N\Phi_\alpha,\qquad A_k=d\mu_N+(dk-1)S,
\end{equation}}
and hence $\Lambda_k$ and $\Phi^{(k)}$ differ only in the factor that multiplies $\Phi_0$. Note that $A_k\leq B_k$ for all $k=1,\ldots,d$ (see Appendix B), and hence $\Phi^{(k)}\leq\Lambda_k$. The same relation follows for the $k$-positive maps from mutually unbiased bases and SIC POVMs \cite{Shi}, which are two special cases of $\Lambda_k$ for $(N,M)=(d+1,d)$ and $(N,M)=(1,d^2)$, respectively. Therefore, $\Phi^{(k)}$ are better than $\Lambda_k$ for the applications in the Schmidt number detection.

\section{Schmidt number witnesses}

Consider a bipartite quantum system with the Hilbert space $\mathcal{H}=\mathcal{H}_A\otimes\mathcal{H}_B\simeq\mathbb{C}^d\otimes\mathbb{C}^d$. 
For bipartite quantum states, the Schmidt number $k$ is an entanglement measure that allows for a hierarchization of quantum states: from $k=1$ for separable states, up to $k=d$ for maximally entangled states \cite{Sperling,TOPICAL}. 
{For any pure bipartite state $|\psi\>\in\mathcal{H}_A\otimes\mathcal{H}_B$, 
there always exist two local orthonormal bases, $\{|e_j\>;\,j=0,\ldots,d-1\}$ in $\mathcal{H}_A$ and $\{|f_j\>;\,j=0,\ldots,d-1\}$ in $\mathcal{H}_B$, such that
\begin{equation}
|\psi\>=\sum_{j=0}^{k-1}\lambda_j|e_j\>\otimes|f_j\>
\end{equation}
with $\lambda_j>0$.
The real numbers $\lambda_j$ are the Schmidt coefficients and satisfy $\sum_{j=0}^{d-1}\lambda_j^2=1$. The integer $k={\rm{SR}}(|\psi\>)$ is the Schmidt rank (SR) of $|\psi\>$ \cite{Peres,Guhne}.}
This approach is then generalized to the mixed states $\rho$ to the Schmidt number \cite{Terhal}
\begin{equation}
{\rm{SN}}(\rho)=\min_{\rho=\sum_jp_j|\psi_j\>\<\psi_j|}\max_j{\rm{SR}}(|\psi_j\>),
\end{equation}
which is the maximal Schmidt rank in a given decomposition of $\rho$, minimized over all possible decompositions. The Schmidt number allows one to introduce a classification of bipartite states. Let
\begin{equation}
S_k=\{\rho:\, {\rm{SN}}(\rho)\leq k\}.
\end{equation}
Then, $S_1\subset S_2\subset\ldots\subset S_d$, where $S_1$ is the set of separable states.

We focus on detecting the Schmidt number of a mixed quantum state using the Schmidt number witness. From definition, $W_k$ is a Schmidt number $k+1$ witness if $\Tr(W_k\rho)\geq 0$ for all $\rho\in S_k$ and there exists a state $\rho$ for which $\Tr(W_k\rho)<0$ \cite{Terhal}. Moreover, the Schmidt number witnesses $W_k$ are related to $k$-positive maps $\Phi_k$ via the Choi-Jamio{\l}kowski isomorphism \cite{Choi,Jamiolkowski},
\begin{equation}\label{W}
W_k=(\oper_d\otimes\Phi_k)[P_+],\qquad P_+=\frac 1d \sum_{m,n=0}^{d-1}|m\>\<n|\otimes|m\>\<n|.
\end{equation}
In what follows, we construct the Schmidt number witnesses from the $k$-positive maps introduced in Proposition 2, so that
\begin{equation}\label{W2}
W_k=\frac{A_k}{d}\mathbb{I}_d\otimes\mathbb{I}_d
+\sum_{\alpha=L+1}^KJ_\alpha-\sum_{\alpha=1}^LJ_\alpha,\qquad
J_\alpha=\sum_{k,\ell=1}^{M_\alpha}\mathcal{O}_{k\ell}^{(\alpha)}
\overline{P}_{\alpha,\ell}\otimes P_{\alpha,k}.
\end{equation}
For $k=1$, one recovers the entanglement witnesses from ref. \cite{GEAM_Pmaps}.

\begin{Example}
Consider the family of isotropic states \cite{Terhal}
\begin{equation}
\rho_F=\frac{1-F}{d^2-1}(\mathbb{I}_d\otimes\mathbb{I}_d-P_+)+FP_+,\qquad 0\leq F\leq 1,
\end{equation}
where $P_+$ is the maximally entangled state given in eq. (\ref{W}). It is known that $\mathrm{SN}(\rho_F)=k+1$ if and only if \cite{Terhal}
\begin{equation}
\frac kd<F\leq \frac{k+1}{d}.
\end{equation}
Let us check the ability of $W_k$ to detect and quantify entanglement of these states.
The expectation value of $J_\alpha$ in the state $\rho_F$ reads
\begin{equation}
\begin{split}
\Tr(J_\alpha\rho_F)&=\frac{1-F}{d^2-1}\Tr(J_\alpha)+\frac{d^2F-1}{d^2-1}\Tr(J_\alpha P_+)=\frac{d(1-F)}{d^2-1}a_\alpha\gamma_\alpha+\frac{d^2F-1}{d(d^2-1)}
\left[S\Tr(\mathcal{O}^{(\alpha)})+a_\alpha^2c_\alpha M_\alpha\right]\\
&=\frac 1d a_\alpha\gamma_\alpha+\frac{d^2F-1}{d(d^2-1)}S
\left[\Tr(\mathcal{O}^{(\alpha)})-1\right],
\end{split}
\end{equation}
where we used the fact that $a_\alpha^2c_\alpha M_\alpha=a_\alpha\gamma_\alpha-S$ for equidistant GEAMs. Therefore, one has
{
\begin{equation}
\Tr(W_k\rho_F)=\frac{S}{d}\left\{dk-1+\frac{d^2F-1}{d^2-1}
\left[\Tr(\widetilde{\mathcal{O}})-K+2L\right]\right\},\qquad \widetilde{\mathcal{O}}=\sum_{\alpha=L+1}^K\mathcal{O}^{(\alpha)}-
\sum_{\alpha=1}^L\mathcal{O}^{(\alpha)}.
\end{equation}}
Recall that $W_k$ witnesses the Schmidt number $k+1$ of $\rho_F$ if $\Tr(W_k\rho_F)<0$. Interestingly, the ability of $W_k$ to detect the Schmidt number of the isotropic state depends on the properties of the rotation matrices but not on the measurement operators.

We analyze two separate cases.
\begin{enumerate}[label=(\it\roman*)]
\item If all $\Tr(\mathcal{O}^{(\alpha)})=0$ (e.g. a permutation matrix), then $\Tr(W_k\rho_F)<0$ for $2L<K$ and
{
\begin{equation}
F>\frac{1}{d^2}\left[1+\frac{(d^2-1)(dk-1)}{K-2L}\right].
\end{equation}}
\item If all $\Tr(\mathcal{O}^{(\alpha)})=h$, $2\leq h\leq \min_\alpha M_\alpha-1$ (partial permutations), then $\Tr(W_k\rho_F)<0$ for $2L>K$ and
{
\begin{equation}
F>\frac{1}{d^2}\left[1+\frac{(d^2-1)(dk-1)}{(2L-K)(h-1)}\right].
\end{equation}}
\end{enumerate}
Therefore, $W_k$ from GEAMs partially detect the Schmidt numbers of $\rho_F$.
\end{Example}

A special class of witnesses follows for mutually unbiased bases in prime dimensions $d$. For a prime $d>2$, there are $d+1$ sets of mutually unbiased bases (MUBs) \cite{Ivanovic,Wiesniak}:
\begin{equation}
\mathcal{B}_\alpha=\left\{U_\alpha|k\>;\,k=0,\ldots,d-1\right\},\qquad \alpha=1,\ldots,d+1,
\end{equation}
where $\omega=\exp(2\pi i/d)$ and the unitary operators
\begin{equation}\label{Ua}
U_{d+1}=\mathbb{I}_d,\qquad U_\alpha=\frac{1}{\sqrt{d}}\sum_{m,n=0}^{d-1}\omega^{\alpha m^2+mn} |m\>\<n|.
\end{equation}
We define the corresponding GEAM via $P_{\alpha,k}=U_\alpha |k\>\<k|U_\alpha^\dagger/(d+1)$. In this case, $N=d+1$, $M_\alpha=d$, and $\gamma_\alpha=1/(d+1)$.

\begin{Example}
Let us take the Werner states \cite{Werner}
\begin{equation}
\rho_\phi=\frac{1}{d(d^2-1)}\left[(d-\phi)\mathbb{I}_d\otimes\mathbb{I}_d+(d\phi-1)\mathbb{F}_d\right],
\qquad -1\leq\phi\leq 1, 
\end{equation}
where $\mathbb{F}_d=\sum_{m,n=0}^{d-1}|m\>\<n|\otimes|n\>\<m|$ is the flip operator. It is known that $\rho_\phi$ is separable for $0\leq\phi\leq 1$ and entangled for $-1\leq\phi<0$.
To detect its Schmidt number with $W_k$, we first compute
\begin{equation}\label{Japhi}
d(d^2-1)\Tr(J_\alpha\rho_\phi)
=\frac{d(d-\phi)}{(d+1)^2}+(d\phi-1)\sum_{k,\ell=0}^{d-1}\mathcal{O}_{k\ell}^{(\alpha)}
\Tr(\overline{P}_{\alpha,k}P_{\alpha,\ell}).
\end{equation}
Now, observe that
\begin{equation}
\overline{P}_{d+1,k}=P_{d+1,k},\qquad 
\overline{P}_{d,k}=P_{d,d-k},\qquad 
\overline{P}_{\alpha,k}=P_{d-\alpha,d-k},\qquad \alpha=1,\ldots,d-1,
\end{equation}
where the second index is calculated modulo $d$. Therefore, eq. (\ref{Japhi}) produces
\begin{equation}
\begin{split}
\Tr(J_{d+1}\rho_\phi)&=\frac{d(d-\phi)+(d\phi-1)\Tr(\mathcal{O}^{(d+1)})}
{d(d^2-1)(d+1)^2},\\
\Tr(J_d\rho_\phi)&=\frac{d(d-\phi)+(d\phi-1)\sum_{k=0}^{d-1}\mathcal{O}_{k,d-k}^{(d)}}
{d(d^2-1)(d+1)^2},\\
\Tr(J_\alpha\rho_\phi)&=\frac{d(d-\phi)+d\phi-1}{d(d^2-1)(d+1)^2},\qquad \alpha=1,\ldots,d-1.
\end{split}
\end{equation}
Note that the choice of $\mathcal{O}^{(\alpha)}$ for $\alpha=1,\ldots,d-1$ has no influence on the Schmidt number detection of $\rho_\phi$.

For $K=L=d+1$, $\mathcal{O}^{(d+1)}=\mathbb{I}_d$, and $\mathcal{O}^{(d)}_{k\ell}=\delta_{\ell,d-k}$, the expectation value of $W_k$ in $\rho_\phi$ is given by
{
\begin{equation}
\Tr(W_k\rho_\phi)=\frac{(d+1)(dk-1)-2(d\phi-1)}{d(d+1)^3}.
\end{equation}}
However, $\Tr(W_k\rho_\phi)\geq 0$ for $-1\leq\phi\leq 1$, which means that $W_k$ do not detect any entanglement of $\rho_\phi$.
\end{Example}

\begin{Example}
There exists a two-parameter generalization of both the isotropic and Werner states in the form of orthogonally invariant quantum states
\begin{equation}
\rho_{A,B}=\frac{1-A-B}{d^2}\mathbb{I}_d\otimes\mathbb{I}_d+AP_++\frac{B}{d}\mathbb{F}_d,
\end{equation}
whose entanglement has been analyzed in ref. \cite{ParkYoun}. Observe that $\rho_{A,B}\geq 0$ if and only if
\begin{equation}
-\frac{1}{d-1}\leq B\leq \frac{d}{d(d+1)-2},\qquad -\frac{1+(d-1)B}{d^2-1}\leq A\leq 1+\min\{(d-1)B,-(d+1)B\}.
\end{equation}
Observe that
\begin{equation}
\rho_{A,B}=\rho_F+\rho_\phi-\frac{1}{d^2}\mathbb{I}_d\otimes\mathbb{I}_d\qquad {\rm for}\qquad A=\frac{d^2F-1}{d^2-1},\qquad B=\frac{d\phi-1}{d^2-1}.
\end{equation}
Now, taking $W_k$ from Example 2 with $K=L=d+1$, $\mathcal{O}^{(\alpha)}=\mathbb{I}_d$ for $\alpha=1,\ldots,d-1,d+1$, and $\mathcal{O}^{(d)}_{k\ell}=\delta_{\ell,d-k}$, we find the sufficient condition for $\rho_{A,B}$ to have the Schmidt number $k+1$,
{
\begin{equation}
\Tr(W_k\rho_{A,B})=\frac{1}{d(d+1)^2}\Big[dk-1-(d^2-d-1)A
-2(d-1)B\Big]<0,
\end{equation}}
or equivalently
{
\begin{equation}
A>\frac{dk-1-2(d-1)B}{d^2-d-1}.
\end{equation}}
\end{Example}


\section{Conclusions}

In this paper, we construct $k$-positive linear maps from a large family of informationally overcomplete measurements given by the equidistant generalized equiangular measurements (GEAMs). After proving an inequality more general than that for the index of coincidence, we derive sufficient conditions for the $k$-positivity of linear maps for any $k=1,\ldots,d$. Interestingly, the entire $k$-dependence is encoded in a scalar parameter. We then compare our maps to the existing $k$-positive maps from subclasses of GEAMs: $(N,M)$-POVMs \cite{SIC-MUB_kpos}, SIC POVMs, and MUBs \cite{Shi}. We prove that our construction leads to a wider class of $k$-positive maps in each of those subclasses. Then, we apply our results to quantify entanglement of bipartite quantum states via the Schmidt number criterion. The Choi-Jamio{\l}kowski isomorphism allows us to characterize the $k+1$ Schmidt number witnesses corresponding to $k$-positive maps. This generalizes the notion of entanglement witnesses, which follow from maps that are only positive. We illustrate our results by detecting the entanglement degree of two-parameter bipartite states in any dimension.

There are several open problems that require further consideration. One could characterize indecomposable witnesses; i.e., those that cannot be rewritten as $W_k=Q_k+R_k^\Gamma$, where $Q_k$ and $R_k$ are $k$-positive maps and $\Gamma=\oper\otimes T$ is the partial transposition. Such witnesses are interesting because they detect quantum states with positive partial transposition (PPT). Other properties that remain to be analyzed are e.g. extremality and optimality. There are several ways to generalize these witnesses: multipartite systems, subsystems of different dimensionality, other families of measurement operators. Another path is to find quantum states whose entanglement degree can be fully determined via our witnesses (that is, using if and only if statements). Finally, in analogy to the mirrored entanglement witnesses in refs. \cite{Bae,Mirror2}, one could characterize mirrored Schmidt number witnesses.

\section{Acknowledgements}

This research was funded in whole or in part by the National Science Centre, Poland, Grant number 2021/43/D/ST2/00102. For the purpose of Open Access, the author has applied a CC-BY public copyright licence to any Author Accepted Manuscript (AAM) version arising from this submission.

%

\bibliography{C:/Users/cynda/OneDrive/Fizyka/bibliography}

\begin{thebibliography}{10}

\bibitem{Masanes}
L.~Masanes.
\newblock All bipartite entangled states are useful for information processing.
\newblock {\em Phys. Rev. Lett.}, 96:150501, 2006.

\bibitem{Ekert}
A.~K. Ekert.
\newblock Quantum cryptography based on bell's theorem.
\newblock {\em Phys. Rev. Lett.}, 67:661, 1991.

\bibitem{Wiesner}
C.~H. Bennett and S.~J. Wiesner.
\newblock Communication via one- and two-particle operators on
  einstein-podolsky-rosen states.
\newblock {\em Phys. Rev. Lett.}, 69:2881--2884, 1992.

\bibitem{BaeDarek}
J.~Bae, D.~Chru{\'s}ci{\'n}ski, and M.~Piani.
\newblock More entanglement implies higher performance in channel
  discrimination tasks.
\newblock {\em Phys. Rev. Lett.}, 122:140404, 2019.

\bibitem{Guhne}
O.~G{\"u}hne and G.~T{\'o}th.
\newblock Entanglement detection.
\newblock {\em Phys. Rep.}, 474:1--75, 2009.

\bibitem{Terhal}
B.~M. Terhal and P.~Horodecki.
\newblock Schmidt number for density matrices.
\newblock {\em Phys. Rev. A}, 61:040301(R), 2000.

\bibitem{Liu4}
S.~Liu, Q.~He, M.~Huber, O.~G{\"u}hne, and G.~Vitagliano.
\newblock Characterizing entanglement dimensionality from randomized
  measurements.
\newblock {\em PRX Quantum}, 4:020324, 2023.

\bibitem{Liu5}
S.~Liu, M.~Fadel, Q.~He, M.~Huber, and G.~Vitagliano.
\newblock Bounding entanglement dimensionality from the covariance matrix.
\newblock {\em Quantum}, 8:1236, 2024.

\bibitem{Johnston4}
N.~Johnston.
\newblock Norm duality and the cross norm criteria for quantum entanglement.
\newblock {\em Linear Multilinear A}, 62:648, 2014.

\bibitem{Klockl}
C.~Klockl and M.~Huber.
\newblock Characterizing multipartite entanglement without shared reference
  frames.
\newblock {\em Phys. Rev. A}, 91:042339, 2015.

\bibitem{Terhal1}
B.~M. Terhal.
\newblock Bell inequalities and the separability criterion.
\newblock {\em Phys. Lett. A}, 271:319, 2000.

\bibitem{Terhal2}
B.~M. Terhal.
\newblock A family of indecomposable positive linear maps based on entangled
  quantum states.
\newblock {\em Linear Algebra Appl.}, 323:61--73, 2001.

\bibitem{Hulpke}
F.~Hulpke, D.~Bruss, M.~Lewenstein, and A.~Sanpera.
\newblock Simplifying schmidt number witnesses via higher-dimensional
  embeddings.
\newblock {\em Quantum Inf. Comput.}, 4:207, 2004.

\bibitem{Sanpera}
A.~Sanpera, D.~Bru{\ss}, and M.~Lewenstein.
\newblock Schmidt-number witnesses and bound entanglement.
\newblock {\em Phys. Rev. A}, 63:050301, 2001.

\bibitem{Wyderka}
N.~Wyderka, G.~Chesi, H.~Kampermann, C.~Macchiavello, and D.~Bru{\ss}.
\newblock Construction of efficient schmidt-number witnesses for
  high-dimensional quantum states.
\newblock {\em Phys. Rev. A}, 107:022431, 2023.

\bibitem{Morelli}
A.~Tavakoli and S.~Morelli.
\newblock Enhanced schmidt-number criteria based on correlation trace norms.
\newblock {\em Phys. Rev. A}, 110:062417, 2024.

\bibitem{Wangs}
Z.~Wang, B.Z. Sun, S.M. Fei, and Z.X. Wang.
\newblock Schmidt number criterion via general symmetric informationally
  complete measurements.
\newblock {\em Quantum Inf. Process.}, 23:401, 2024.

\bibitem{Schmidt_NM}
H.-F. Wang and S.-M. Fei.
\newblock Estimating the schmidt numbers of quantum states via symmetric
  measurements.
\newblock {\em Ann. der Phys.}, 537:e00259, 2025.

\bibitem{GEAM_coherence}
K.~Siudzi\'{n}ska.
\newblock Measures from conical 2-designs depend only on two constants.
\newblock {\em J. Phys. A: Math. Theor.}, 58:375302, 2025.

\bibitem{Shi}
X.~Shi.
\newblock Families of schmidt-number witnesses for high dimensional quantum
  states.
\newblock {\em Commun. Theor. Phys.}, 76:085103, 2024.

\bibitem{SIC-MUB_kpos}
X.-Q. Mu, H.-F. Wang, and S.-M. Fei.
\newblock A note on schmidt-number witnesses based on symmetric measurements.
\newblock {\em Laser Phys. Lett}, 22:115208, 2025.

\bibitem{MUBs}
D.~Chru{\'s}ci{\'n}ski, G.~Sarbicki, and F.~A. Wudarski.
\newblock Entanglement witnesses from mutually unbiased bases.
\newblock {\em Phys. Rev. A}, 97(12):032318, 2018.

\bibitem{EW-2MUB}
K.~Wang and Z.-J. Zheng.
\newblock Constructing entanglement witnesses from two mutually unbiased bases.
\newblock {\em Int. J. Theor. Phys.}, 60:274--283, 2021.

\bibitem{bound_ent}
J.~Bae, A.~Bera, D.~Chru{\'s}ci{\'n}ski, B.~C. Hiesmayr, and D.~McNulty.
\newblock How many measurements are needed to detect bound entangled states?
\newblock {\em J. Phys. A: Math. Theor.}, 55:505303, 2022.

\bibitem{Li}
T.~Li, L.-M. Lai, S.-M. Fei, and Z.-X. Wang.
\newblock Mutually unbiased measurement based entanglement witnesses.
\newblock {\em Int. J. Theor. Phys.}, 58:3973--3985, 2019.

\bibitem{P_maps}
K.~Siudzi{\'n}ska and D.~Chru{\'s}ci{\'n}ski.
\newblock Entanglement witnesses from mutually unbiased measurements.
\newblock {\em Sci. Rep.}, 11:22988, 2021.

\bibitem{MUM_purity}
M.~Salehi, S.~J. Akhtarshenas, M.~Sarbishaei, and H.~Jaghouri.
\newblock Mutually unbiased measurements with arbitrary purity.
\newblock {\em Quantum Inf. Process.}, 20:401, 2021.

\bibitem{EW-SIC}
T.~Li, L.-M. Lai, D.-F. Liang, S.-M. Fei, and Z.-X. Wang.
\newblock Entanglement witnesses based on symmetric informationally complete
  measurements.
\newblock {\em Int. J. Theor. Phys.}, 59:3549--3557, 2020.

\bibitem{SIC-MUB}
K.~Siudzi{\'n}ska.
\newblock All classes of informationally complete symmetric measurements in
  finite dimensions.
\newblock {\em Phys. Rev. A}, 105:042209, 2022.

\bibitem{W_ETS}
X.~Shi.
\newblock The entanglement criteria based on equiangular tight frames.
\newblock {\em J. Phys. A: Math. Theor.}, 57:075302, 2024.

\bibitem{Rastegin_EW}
A.~E. Rastegin.
\newblock Separability criteria and entanglement witnesses from mutually
  unbiased equiangular tight frames.
\newblock {\em Commun. Theor. Phys.}, 78:025104, 2026.

\bibitem{GEAM}
K.~Siudzi{\'n}ska.
\newblock Informationally overcomplete measurements from generalized
  equiangular tight frames.
\newblock {\em J. Phys. A: Math. Theor.}, 57:335302, 2024.

\bibitem{EOM22}
L.~Feng and S.~Luo.
\newblock Equioverlapping measurements.
\newblock {\em Phys. Lett. A}, 445:128243, 2022.

\bibitem{Strohmer}
T.~Strohmer.
\newblock A note on equiangular tight frames.
\newblock {\em Linear Algebra Appl.}, 429:326, 2008.

\bibitem{Gour}
A.~Kalev and G.~Gour.
\newblock Construction of all general symmetric informationally complete
  measurements.
\newblock {\em J. Phys. A: Math. Theor.}, 47:335302, 2014.

\bibitem{Kalev}
A.~Kalev and G.~Gour.
\newblock Mutually unbiased measurements in finite dimensions.
\newblock {\em New J. Phys.}, 16:053038, 2014.

\bibitem{SIC-MUB_general}
K.~Siudzi{\'n}ska.
\newblock How much symmetry do symmetric measurements need for efficient
  operational applications?
\newblock {\em J. Phys. A: Math. Theor.}, 57:355301, 2024.

\bibitem{conical}
K.~Siudzi\'{n}ska.
\newblock Equivalence relations between conical 2-designs and mutually unbiased
  generalized equiangular tight frames.
\newblock {\em Mathematics}, 14:128, 2026.

\bibitem{Graydon}
M.~A. Graydon and D.~M. Appleby.
\newblock Quantum conical designs.
\newblock {\em J. Phys. A: Math. Theor.}, 49:085301, 2016.

\bibitem{Graydon2}
M.~A. Graydon and D.~M. Appleby.
\newblock Entanglement and designs.
\newblock {\em J. Phys. A: Math. Theor.}, 49:33LT02, 2016.

\bibitem{Rastegin5}
A.~E. Rastegin.
\newblock Uncertainty relations for mubs and sic-povms in terms of generalized
  entropies.
\newblock {\em Eur. Phys. J. D}, 67:269, 2013.

\bibitem{GEAM_Pmaps}
K.~Siudzi\'{n}ska.
\newblock Entanglement witnesses and separability criteria based on generalized
  equiangular tight frames.
\newblock {\em Sci. Rep.}, 15:29890, 2025.

\bibitem{Marciniak}
T.~M{\l}ynik, H.~Osaka, and M.~Marciniak.
\newblock Characterization of $k$-positive maps.
\newblock {\em Commun. Math. Phys.}, 406:62, 2025.

\bibitem{Ende2}
F.~vom Ende, S.~Khatri, and S.~Denisov.
\newblock $k$-positive maps: New characterizations and a generation method.
\newblock {\em Open Syst. Inf. Dyn.}, 32:2550015, 2025.

\bibitem{Stormer}
E.~St{\o}rmer.
\newblock {\em Positive linear maps of operator algebras}.
\newblock Springer-Verlag, Berlin, 2013.

\bibitem{Stormer2}
E.~St{\o}rmer.
\newblock Positive linear maps of operator algebras.
\newblock {\em Acta Math.}, 110:233--278, 1963.

\bibitem{Paulsen}
V.~Paulsen.
\newblock {\em Completely Bounded Maps and Operator Algebras}.
\newblock Cambridge University Press, Cambridge, 2003.

\bibitem{EBC}
M.~Horodecki, P.~W. Shor, and M.~B. Ruskai.
\newblock Entanglement breaking channels.
\newblock {\em Rev. Math. Phys.}, 15:629--641, 2003.

\bibitem{Mehta}
M.~L. Mehta.
\newblock {\em Matrix theory : selected topics and useful results}.
\newblock Delhi Hindustan Publishing Corporation, Delhi, 1989.

\bibitem{Zyczkowski}
I.~Bengtsson and K.~\.{Z}yczkowski.
\newblock {\em Geometry of Quantum States: An Introduction to Quantum
  Entanglement}.
\newblock Cambridge University Press, Cambridge, 2007.

\bibitem{Sperling}
J.~Sperling and W.~Vogel.
\newblock The schmidt number as a universal entanglement measure.
\newblock {\em Phys. Scr.}, 83:045002, 2011.

\bibitem{TOPICAL}
D.~Chru\'sci\'nski and G.~Sarbicki.
\newblock Entanglement witnesses: construction, analysis and classification.
\newblock {\em J. Phys. A: Math. Theor.}, 47:483001, 2014.

\bibitem{Peres}
A.~Peres.
\newblock {\em Quantum Theory: Concepts and Methods}.
\newblock Springer, 1995.

\bibitem{Choi}
M.-D. Choi.
\newblock Completely positive linear maps on complex matrices.
\newblock {\em Linear Algebra Appl.}, 10:285--290, 1975.

\bibitem{Jamiolkowski}
A.~Jamio{\l}kowski.
\newblock Linear transformations which preserve trace and positive
  semidefiniteness of operators.
\newblock {\em Rep. Math. Phys.}, 3:275--278, 1972.

\bibitem{Ivanovic}
I.~D. Ivanovi{\'{c}}.
\newblock Geometrical description of quantal state determination.
\newblock {\em J. Phys. A: Math. Theor.}, 14:3241, 1981.

\bibitem{Wiesniak}
M.~Wie\'{s}niak, T.~Paterek, and A.~Zeilinger.
\newblock Entanglement in mutually unbiased bases.
\newblock {\em New J. Phys.}, 13:053047, 2011.

\bibitem{Werner}
R.~Werner.
\newblock Quantum states with einstein-podolsky-rosen correlations admitting a
  hidden-variable model.
\newblock {\em Phys. Rev. A}, 40:4277, 1989.

\bibitem{ParkYoun}
S.-J. Park and S.-G. Youn.
\newblock $k$-positivity and schmidt number under orthogonal group symmetries.
\newblock {\em Quantum Inf. Process.}, 23:162, 2024.

\bibitem{Bae}
J.~Bae, D.~Chru{\'s}ci{\'n}ski, and B.~C. Hiesmayr.
\newblock Mirrored entanglement witnesses.
\newblock {\em npj Quantum Inf.}, 6:15, 2020.

\bibitem{Mirror2}
D.~Chru\'{s}ci\'{n}ski, A.~Bera, J.~Bae, and B.~C. Hiesmayr.
\newblock A mirrored pair of optimal non-decomposable entanglement witnesses
  for two qudits does exist.
\newblock {\em Sci. Rep.}, 15:28205, 2025.

\bibitem{Rastegin2}
A.~E. Rastegin.
\newblock On uncertainty relations and entanglement detection with mutually
  unbiased measurements.
\newblock {\em Open Sys. Inf. Dyn.}, 22:1550005, 2015.

\end{thebibliography}
\bibliographystyle{unsrt}

\appendix

\section{Proof of Lemma 1}\label{AppA}

The proof is analogical to that in ref. \cite{GEAM_Pmaps} for the partial index of coincidence and follows the methodology from ref. \cite{Rastegin2}. We start by expanding an arbitrary linear operator on $\mathcal{H}$,
\begin{equation}\label{rho}
X=\frac{x_0}{d}\mathbb{I}_d+\sum_{\alpha=1}^N\sum_{k=1}^{M_\alpha}r_{\alpha,k}H_{\alpha,k},
\end{equation}
in the frame of the traceless operators $H_{\alpha,k}$ defined in eq. (\ref{H}) and the identity $\mathbb{I}_d$. Note that the parameters $r_{\alpha,k}$ are complex and $\Tr(X)=x_0$. Next, from the properties of $H_{\alpha,k}$,
\begin{equation}\label{HHH}
\begin{split}
\Tr(H_{\alpha,k}^2)&=(M_\alpha-1)(\sqrt{M_\alpha}+1)^2,\\
\Tr(H_{\alpha,k}H_{\alpha,\ell})&=-(\sqrt{M_\alpha}+1)^2,\qquad \ell\neq k,\\
\Tr(H_{\alpha,k}H_{\beta,\ell})&=0,\qquad \beta\neq\alpha.
\end{split}
\end{equation}
it follows that
\begin{equation}
\Tr(X P_{\alpha,k})
= \frac{a_\alpha x_0}{d} +\tau_\alpha\Tr(X H_{\alpha,k})
= \frac{a_\alpha x_0}{d}
+\tau_\alpha\sum_{\ell=1}^{M_\alpha}r_{\alpha,\ell}\Tr(H_{\alpha,k}H_{\alpha,\ell})
= \frac{a_\alpha x_0}{d}+\tau_\alpha(\sqrt{M_\alpha}+1)^2(M_\alpha r_{\alpha,k}-r_\alpha),
\end{equation}
where $r_\alpha=\sum_{k=1}^{M_\alpha}r_{\alpha,k}$. 
Now, if we take the module squared and sum it over the entire range of index $k$, we arrive at
\begin{equation}\label{sum}
\sum_{k=1}^{M_\alpha}|\Tr(X P_{\alpha,k})|^2
=\frac{a_\alpha^2M_\alpha}{d^2} |x_0|^2+\tau_\alpha^2M_\alpha(\sqrt{M_\alpha}+1)^4
\left(M_\alpha\sum_{k=1}^{M_\alpha} |r_{\alpha,k}|^2-|r_\alpha|^2\right).
\end{equation}
From the Cauchy-Schwarz inequality, we know that $|r_\alpha|^2\leq M_\alpha\sum_{k=1}^{M_\alpha} |r_{\alpha,k}|^2$.
This proves the positivity of the second term on the right hand-side of eq. (\ref{sum}).
It remains to calculate the sum over an incomplete set of indices $\alpha=1,\ldots,L\leq N$, which reads
\begin{equation}\label{almost}
\begin{split}
\sum_{\alpha=1}^L\sum_{k=1}^{M_\alpha}|\Tr(X P_{\alpha,k})|^2
&=\frac{|x_0|^2}{d^2}\sum_{\alpha=1}^L a_\alpha^2 M_\alpha +\sum_{\alpha=1}^L\tau_\alpha^2M_\alpha(\sqrt{M_\alpha}+1)^4
\left(M_\alpha\sum_{k=1}^{M_\alpha} |r_{\alpha,k}|^2-|r_\alpha|^2\right)\\
&=\frac{|x_0|^2}{d^2}\sum_{\alpha=1}^L a_\alpha^2 M_\alpha +\sum_{\alpha=1}^La_\alpha^2(b_\alpha-c_\alpha)(\sqrt{M_\alpha}+1)^2
\left(M_\alpha\sum_{k=1}^{M_\alpha} |r_{\alpha,k}|^2-|r_\alpha|^2\right)\\
&=\frac{|x_0|^2}{d^2}\sum_{\alpha=1}^L a_\alpha^2 M_\alpha +S\sum_{\alpha=1}^L(\sqrt{M_\alpha}+1)^2
\left(M_\alpha\sum_{k=1}^{M_\alpha}|r_{\alpha,k}|^2-|r_\alpha|^2\right).
\end{split}
\end{equation}
On the other hand, going back to eq. (\ref{rho}), we show that
\begin{equation}\label{rho2}
\Tr(X^\dagger X)=\frac{|x_0|^2}{d} + \sum_{\alpha=1}^N(\sqrt{M_\alpha}+1)^2
\left(M_\alpha\sum_{k=1}^{M_\alpha} |r_{\alpha,k}|^2-|r_\alpha|^2\right),
\end{equation}
where we have applied the properties of $H_{\alpha,k}$ from eq. (\ref{HHH}) to simplify
\begin{equation}
\sum_{k,\ell=1}^{M_\alpha}r_{\alpha,k}\overline{r}_{\alpha,\ell}\Tr(H_{\alpha,k}H_{\alpha,\ell})
=(\sqrt{M_\alpha}+1)^2\left(M_\alpha\sum_{k=1}^{M_\alpha}|r_{\alpha,k}|^2-|r_\alpha|^2\right).
\end{equation}
Finally, if we compare eqs. (\ref{almost}) and (\ref{rho2}), it turns out that
\begin{equation}
\begin{split}
\sum_{\alpha=1}^L\sum_{k=1}^{M_\alpha}|\Tr(X P_{\alpha,k})|^2
&=\frac{|x_0|^2}{d^2} \sum_{\alpha=1}^L a_\alpha^2 M_\alpha 
+S\sum_{\alpha=1}^L(\sqrt{M_\alpha}+1)^2
\left(M_\alpha\sum_{k=1}^{M_\alpha} |r_{\alpha,k}|^2-|r_\alpha|^2\right)\\
&\leq\frac{|x_0|^2}{d^2} \sum_{\alpha=1}^L a_\alpha^2 M_\alpha 
+S\sum_{\alpha=1}^N(\sqrt{M_\alpha}+1)^2
\left(M_\alpha\sum_{k=1}^{M_\alpha} |r_{\alpha,k}|^2-|r_\alpha|^2\right)\\
&=|\Tr(X)|^2\mu_L+S\left[\Tr(X^\dagger X)-\frac{|\Tr(X)|^2}{d}\right]\\
&=S\Tr(X^\dagger X)+|\Tr(X)|^2\left(\mu_L-\frac{S}{d}\right).
\end{split}
\end{equation}

\section{Proof that $\Phi^{(k)}\leq \Lambda_k$}

Here, we prove that $A_k\leq B_k$, where
{\begin{equation}
A_k=d\mu_N+(dk-1)S
\end{equation}}
and
\begin{equation}
B_k=N^2\mu_N\left[d+(dk-1)\sqrt{S\left(d+(d^2-1)\frac{S}{\mu_N}\right)}\right].
\end{equation}
Note that $N\geq 1$, and hence $d\mu_N\leq dN^2\mu_N$. Thus, it is enough to show that
{\begin{equation}
(dk-1)S\leq N^2\mu_N(dk-1)\sqrt{S\left(d+(d^2-1)\frac{S}{\mu_N}\right)}.
\end{equation}}
Both terms are positive, so we take the square without changing the sign;
{\begin{equation}
(dk-1)^2S^2\leq N^4\mu_N(dk-1)^2S\left(d\mu_N+(d^2-1)S\right).
\end{equation}}
This further simplifies to
{\begin{equation}\label{ineq}
S\leq N^4\mu_N\left(d\mu_N+(d^2-1)S\right).
\end{equation}}
We prove a stronger relation,
{\begin{equation}
S\leq N^4\mu_N(d^2-1)S,
\end{equation}}
which reduces to
{\begin{equation}
M\leq N^3(d^2-1),
\end{equation}}
due to $\mu_N=1/(MN)$ {for $(N,M)$-POVMs. Finally, observe that}
\begin{equation}
M\leq N^3(d^2-1)\qquad\Longleftrightarrow\qquad 
N^4\geq 1+\frac{1}{M-1},
\end{equation}
where we used the identity $N(M-1)=d^2-1$ {that holds for all informationally complete $(N,M)$-POVMs.} As $M\geq 2$, the above inequality always holds for $N\geq 2$.

{If $N=1$, then $M=d^2$ and eq. (\ref{ineq}) is equivalent to
\begin{equation}\label{eee}
S\leq\frac 1d.
\end{equation}
Finally, we compare this with the admissible range of $S$ from eq. (\ref{Srange}). For the $(1,d^2)$-POVMs, this range is simply
\begin{equation}
0<S\leq \frac{1}{d(d+1)},
\end{equation}
and therefore eq. (\ref{eee}) is always satisfied.}

\end{document}